\documentclass[aps,prl,twocolumn,superscriptaddress,longbibliography]{revtex4}

\usepackage{amsmath}
\usepackage{times}
\usepackage{hyperref}
\usepackage{cleveref}
\usepackage{autonum}
\usepackage{graphicx}
\usepackage{dcolumn}
\usepackage{enumerate,amsthm,amssymb,color}
\usepackage{bbm}
\usepackage{bm}
\usepackage{color}
\usepackage{tikz}
\usepackage[french, english]{babel}
\usepackage{synttree}
\usepackage{multirow}
\usepackage{xfrac}
\usepackage{fixltx2e}
\usepackage{lipsum}

\newcommand{\ket}[1]{|#1\rangle}

\newcommand{\COMMENT}[1]{}

\def\defeq{\mathrel{\mathop:}=} 

\definecolor{violet}{HTML}{53257F} 
\definecolor{green}{HTML}{257a7f}
\definecolor{brown}{HTML}{852e29}

\newtheorem{theorem}{Theorem}

\newtheorem{ring}{Ring}
\newtheorem{linie}{Line}
\newtheorem{define}{Definition}
\newtheorem{lemma}{Lemma}

\newtheorem*{observation*}{Observation}

\begin{document}
	
	\title{Limitations of nearest-neighbour quantum networks}
	
	\author{F.\ Hahn}\affiliation{Dahlem Center for Complex Quantum Systems, Freie Universit{\"a}t Berlin, 14195 Berlin, Germany}\affiliation{Electrical Engineering and Computer Science Department, Technische Universit{\"a}t Berlin, 10587 Berlin, Germany}
	\author{A.\ Dahlberg}\affiliation{QuTech, TU Delft, 2628CJ Delft, The Netherlands}
	\author{J.\ Eisert}\affiliation{Dahlem Center for Complex Quantum Systems, Freie Universit{\"a}t Berlin, 14195 Berlin, Germany}
	\affiliation{Fraunhofer Heinrich Hertz Institute, 10587 Berlin, Germany}
	\author{A.\ Pappa}\affiliation{Electrical Engineering and Computer Science Department, Technische Universit{\"a}t Berlin, 10587 Berlin, Germany}\affiliation{Fraunhofer Institute for Open Communication Systems - FOKUS, 10589 Berlin, Germany}

	\date{\today}
	
	\begin{abstract}
Quantum communication research has in recent years shifted to include multi-partite networks for which questions of quantum network routing naturally emerge. To understand the potential for multi-partite routing, we focus on the most promising architectures for future quantum networks – those connecting nodes close to
each other. Nearest-neighbour networks such as rings, lines, and grids, have been studied under different communication scenarios to facilitate the sharing of quantum resources, especially in the presence of bottlenecks.
We here analyze the potential of nearest-neighbour quantum networks and identify some serious limitations, by demonstrating that rings and lines cannot overcome common bottleneck communication problems.
	\end{abstract}
	
	\maketitle
	
	\section{Introduction}
	The idea of communication in
	multi-partite
	quantum networks beyond point-to-point quantum key distribution has gained substantial momentum in recent years. Such applications
	suggest that it should be possible to distribute quantum data stably over a wide area
	\cite{QuantumInternet,QuantumInternetWehner, Guha, VanMeter,SMIKW16,hahn2018quantum}. In this context, new
	and challenging questions of how to pursue routing and
	quantum network coding arise \cite{raussendorf_measurement-based_2003, debeaudrap14, EppingA}.
	
	A key multi-partite feature is the promise
	to solve communication bottlenecks in quantum networks. The best known and practically motivated example is perhaps the butterfly network \cite{Butterfly, hayashi_quantum_2007, meignant2021classical}, where two pairs of nodes intend to send quantum messages between them, bypassing the existing bottleneck in the network. To understand the potential and limitations of quantum network coding, a sound level of abstraction to discuss such delicate challenges is to capture the quantum state as a \emph{graph state} \cite{Hein04,Hein06,hahn2018quantum,dahlberg2020transforming,dahlberg2022complexity,dahlberg2020}. 
	
	But what is the underlying property that enables bypassing existing bottlenecks in the network? In the case of the butterfly cluster state, an instance of a graph state, $X$-measurements on the two `additional' nodes lead to the creation of two crossing maximally-entangled
    pairs that enable further quantum communication via teleportation. This is equivalent to a more widely-applicable technique that uses graph transformations called \emph{local complementations} (LCs). As shown in Ref.~\cite{hahn2018quantum}, LCs can reveal 'hidden' properties of the shared graph state and allow us to optimise the quantum resources available.
    
    It is generally more common to have a bottleneck in sparse networks, where the connectivity between the different nodes is limited. A particular case of sparse networks that is of great importance for quantum communication is that of nearest-neighbour architectures. Such networks allow quantum information to travel only over short distances and therefore aim to minimise the noise and losses in the transmission. The butterfly network is one of the smallest instances of a grid network, while other common nearest-neighbour architectures are lines and rings. 
    However, except for the case of the butterfly and related examples \cite{Butterfly, kobayashi2010perfect}, not much is known to date about what is possible in this type of network architectures. 
    
    In this work, we examine in detail whether we can extend the prominent
    example of the butterfly network to other nearest-neighbour architectures.
    We specifically ask whether simultaneous communication of two pairs of nodes is possible in bottleneck scenarios when the underlying architecture is a ring or a line. We conclude that these nearest-neighbour networks are unsuitable for bypassing bottlenecks, and a long-distance communication link is required. Finally, our techniques can find application to more general network topologies and bottleneck settings.
	
	\section{Methods}
	Throughout this work, we stay in the framework where quantum systems held by the respective parties of the network are qubits and are -- on a level of abstraction -- seen as being in pure quantum states. The connectivity pattern of the network is captured by a suitable graph, following the mindset of
	Refs.~\cite{hahn2018quantum,dahlberg2020transforming,dahlberg2022complexity,dahlberg2020}.
	A graph $G=(V,E)$ consists of a finite set of vertices $V\subsetneq \mathbb{N}$ and of edges $E\subseteq V \times V$. 
	We consider \emph{simple} graphs that neither contain edges connecting a vertex to itself nor multiple edges between the same pair of vertices.
	The set of vertices sharing an edge with vertex $v$ is called its \emph{neighbourhood} and denoted as $N_v$. The graph's adjacency matrix is
	\vspace{-0.1in}
	\begin{equation}
	\left(\Gamma_G\right)_{i,j}\defeq\begin{cases} 
	1 & \text{if }(i,j)\in E\\
	0 & \text{if }(i,j)\not\in E.
	\end{cases}
	\end{equation}
	A \emph{graph state} vector 
	\cite{Hein04}
	$\ket{G}$ is defined by $|V|$ qubits in 
	$\ket{+}\defeq (\ket{0}+\ket{1})/\sqrt{2}$ entangled via $CZ$ gates for each edge,
	\begin{equation}
	\ket{G}\defeq\prod_{(i,j)\in E} CZ_{i,j} \ket{+}^{\otimes |V|}.
	\end{equation}
	Physically important are \emph{local Clifford} operations, i.e.~local unitaries from the 	single-qubit Clifford group.
	Such operations on a graph state are interestingly reflected
	by simple transformations of the respective
	graph, namely local complementations 	\cite{Hein04,VandenNest1}. We define the following.
	
	\begin{define}[Local complementation]
		A graph $G=(V,E)$ and vertex $v\in V$ define a locally complemented graph $\tau_v(G)$ with adjacency matrix 
		\begin{equation}
		\Gamma_{\tau_v(G)}\defeq\Gamma_G+\Theta_v \mod 2,
		\end{equation}
		where $\Theta_v$ is the complete graph of the neighbourhood $N_v$.
	\end{define}
	The graph state vector obtained from local complementation with respect to node $v$ of graph $G$, is defined by $\ket{\tau_v(G)}\defeq U_v^\tau\ket{G}$, where $U_v^\tau\defeq ({-iX_v})^{1/2} ({i Z_{N_v}})^{1/2}$. 
	Deciding whether or not
	two graphs can be transformed into each other via sequential LCs is possible in polynomial time  \cite{VandenNest2}.
	As we only consider local Clifford operations and Pauli measurements, the resulting states remain graph states and can be described in terms of the pre-measurement graph with LCs and vertex deletions \cite{Hein04,Hein06}.
    \begin{figure*}[t]
    \centering
    	\includegraphics[width=1.0\textwidth,
    	trim={5.9cm 12.5cm 5.9cm 10.5cm}, 
    	clip]{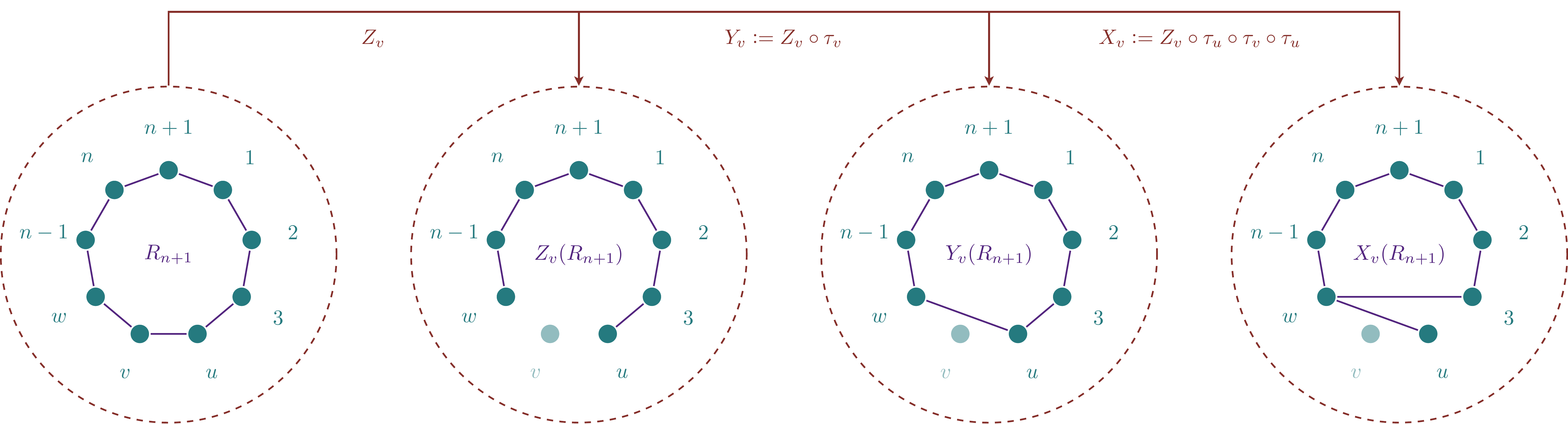}
    		\caption{{\bf Measurements on a ring.}
    		From left to right: Ring $R_{n+1}$ with $n+1$ vertices, node $v$ is measured in the $Z$-, $Y$- and $X$-basis.}
		    \label{fig:X_n+1}
		    \label{fig:Y_n+1}
		    \label{fig:Z_n+1}
		    \label{fig:n+1}
    		\label{fig:RingXYZ}
    \end{figure*}
    
    We hence define local Pauli measurements $P_v$ with respect to this graph action, i.e., $P_v \in \{X_v,Y_v,Z_v\}$ maps a graph with $n$ vertices to one with $n-1$ by removing $v$. Note that local Pauli measurements on $\ket{G}$ result in a different graph state up to local unitary corrections (cf.\ Prop.\ 7 in Ref.~\cite{Hein06}). We will omit these corrections for the sake of clarity.
    
	\begin{define}[Pauli measurements]\label{def:Paulimeasurements}
	The graph action of $Z_v$ is $Z_v(G)\defeq(\tilde{V},E \cap \tilde{V}\times\tilde{V})$ with $\tilde{V}\defeq V\setminus\{v\}$, that is, 
	deleting the row and column of $v$ from $\Gamma_{G}$ gives $\Gamma_{Z_v(G)}$.
	With local complementations we further have
	$
	Y_v(G)\defeq
	Z_v\circ \tau_v(G)
	$
	and
	$
	X_v(G)\defeq
	Z_v\circ\tau_w\circ \tau_v \circ\tau_w(G),
	$
	where $w\in N_v$.
	\end{define}

	A graph $H$ that can be obtained from a graph $G$ via a sequence of local complementations and vertex deletions is called a \emph{vertex-minor} of $G$ \cite{VertexMinor}.
	We will denote this as $H<G$ and call a graph \emph{$v$-minor} of another graph, if $v$ is
	the single vertex that has been
	deleted, e.g., $X_v(G),Y_v(G)$ and $Z_v(G)$ are $v$-minors of $G$. Deciding whether or not
	a graph $H$ is a vertex-minor of $G$ on some subset of vertices is \emph{NP-complete} \cite{dahlberg2022complexity}. Likewise, extracting a set of Bell pairs on a fixed set of vertices from a general graph is NP-complete \cite{dahlberg2020transforming}. Exactly because solving general problems in quantum network routing is provably hard, we focus on impossibility results for widely used network architectures: rings and lines. Due to their symmetry, we consider the former first. Theorems \ref{thm:nocrossing} and \ref{thm:nocrossingline} eliminate specific instances of -- what we expect to be -- commonly-encountered problems in future quantum networks.

\section{Ring graphs}

We first consider graph state vectors $|R_n\rangle$ corresponding to ring graphs, i.e.~ $R_n\defeq(V_n,E_n)$ with $V_n\defeq\{1,\ldots,n\}$ and $E_n\defeq\{(1,2),(2,3),\ldots,(n-1,n),(n,1)\}$ (Fig.~\ref{fig:n+1}). Our goal is to obtain two maximally entangled pairs between qubits $\{a_1,a_2\}$ and $\{b_1,b_2\}$ via local Clifford operations and Pauli measurements.
That is, we want to determine if the most simple graph with two connected components 
\begin{equation}
K_2\cup K_2 \defeq \left(\{a_1,a_2,b_1,b_2\},\{(a_1,a_2),(b_1,b_2)\}\right)
\end{equation}
is a vertex-minor of $R_n$. Without loss of generality we can restrict to $a_1<a_2$, $b_1<b_2$ and set $a_1=1$. In order to show that it is not possible to achieve our goal if $a_1<b_1<a_2<b_2$, we will make use of the following lemmas.

\begin{lemma}[Theorem 3.1 \cite{dahlberg2020}]\label{lem:XYZ}
Let $G$ and $H$ be two graphs and $(v_1,v_2,\ldots,v_k)$ be an ordered tuple of vertices that contains each element of $V_G \setminus V_{H}$ exactly once. We define the corresponding set of possible Pauli operations as
\begin{equation}
\mathcal{P}_{(v_1,v_2,\ldots,v_k)} \defeq \{P_{v_k}\circ P_{v_{k-1}}\circ \cdots \circ P_{v_1}| P_{v}\in \{X_v,Y_v,Z_v\}\}.
\end{equation}
Then $H$ is a vertex-minor of $G$ if and only if there exists an operation $P\in \mathcal{P}_{(v_1,v_2,\ldots,v_k)}$ such that $H$ can be obtained from $P(G)$ via a sequence of local complementations.
\end{lemma}

It will also be useful to single out two specific types of vertices, namely \emph{leaves} and \emph{axils}.

\begin{define}[Leaf and axil]\label{def:leafaxil}
A leaf is a vertex with degree one. An axil is the unique neighbour of a leaf.
\end{define}

For leaves and axils, we have the following lemma regarding the relevant vertex-minors.

\begin{lemma}[Theorem 2.7 \cite{Wehner18}]\label{lem:leafaxil}
Let $G$ and $H$ be graphs and $v$ be a vertex in $V_G$ but not in $V_H$. Then it holds that:
\begin{enumerate}
\item [(a)] If $v$ is a leaf: $H<G \Leftrightarrow H<G \backslash v$.
\item [(b)] If $v$ is an axil: $H<G \Leftrightarrow H<\tau_w\circ\tau_v(G) \backslash v$, where $w$ is the leaf associated with $v$.
\end{enumerate}
\end{lemma}

Note that $(b)$ follows from $(a)$ as leaf and associated axil can be transformed into each other via local complementation. With Lemmas~\ref{lem:XYZ}  and~\ref{lem:leafaxil} we can prove our no-go results (Theorems~\ref{thm:nocrossing} and \ref{thm:nocrossingline}). In combination with Theorem~\ref{thm:LCfoliage}, we provide a tool that can find application to more general network architectures that are not limited to nearest-neighbour ones.

\begin{theorem}[No crossing on a ring]\label{thm:nocrossing}
It is not possible to extract two maximally entangled pairs from $|R_n\rangle$ if $a_1=1<b_1<a_2<b_2$ for any $n\in \mathbb{N}$ with local Clifford operations, local Pauli measurements and classical communication.
\end{theorem}

\begin{proof}
The proof works by induction. The base case is trivial for  $1 \leq n\leq 4$. For $n=5$ and $n=6$ it can be derived by Propositions 1 and 2 in Ref.~\cite{hahn2018quantum}, since ring graphs have a bottleneck with respect to communication requests of the type  $a_1<b_1<a_2<b_2$. 

For the inductive step, we now assume that Theorem~\ref{thm:nocrossing} holds up to a given $n$. We can then build an argument on three case distinctions to show the same follows for $n+1$. In order to see this, note that $K_2 \cup K_2$ is a vertex-minor of $R_{n+1}$ if and only if it is a vertex-minor of at least one of $X(R_{n+1})$, $Y(R_{n+1})$ or $Z(R_{n+1})$ (Lemma~\ref{lem:XYZ}). In the following we will show that it is not a vertex-minor of any of them. More specifically, any $v$-minor $H$ of $R_{n+1}$ is, according to Lemma \ref{lem:XYZ}, equivalent to either $X_v(R_{n+1}),Y_v(R_{n+1})$ or $Z_v(R_{n+1})$ via a sequence of local complementations. As the vertex-minor relationship is inherited via local complementations, $K_2 \cup K_2$ can only be a vertex-minor of $H$ if it is a vertex-minor of at least one of $X(R_{n+1})$, $Y(R_{n+1})$ or $Z(R_{n+1})$. It is therefore enough to show that $K_2 \cup K_2\not < P_v(R_{n+1})$ for all $P_v \in \{X_v,Y_v,Z_v\}$, where we have $v \notin \{a_1,a_2,b_1,b_2\}$. We consider the three cases separately.

\begin{ring}[]\label{claim:1}
$K_2 \cup K_2\not < Z_v(R_{n+1})$.
\end{ring}

\noindent $Z_v(R_{n+1})$ is the line graph $L_n$ with $n$ vertices as depicted in Fig.~\ref{fig:Z_n+1}. With Theorem~\ref{thm:nocrossingline} we find that maximally entangled pairs $(a_1,a_2)$ and $(b_1,b_2)$ can not be extracted from the corresponding graph state vector $|L_n\rangle$, i.e., $K_2 \cup K_2\not < Z_v(R_{n+1})$.

\begin{ring}[]\label{claim:2}
$K_2 \cup K_2\not < Y_v(R_{n+1})$.
\end{ring}
 
\noindent Since $Y_v(R_{n+1})=R_n$ (see also Fig.~\ref{fig:Y_n+1}), we can use our induction hypothesis to infer that $|Y_v(R_{n+1})\rangle$ does not allow for the extraction of  maximally entangled pairs $(a_1,a_2)$ and $(b_1,b_2)$, that is, we have $K_2 \cup K_2\not < Y_v(R_{n+1})$.

\begin{ring}[]\label{claim:3}
$K_2 \cup K_2\not < X_v(R_{n+1})$.
\end{ring}
 
\noindent $X_v(R_{n+1})$ is the ring graph $R_n$ with an additional leaf as depicted in Fig.~\ref{fig:X_n+1}. The former neighbours of $v$ within the graph $R_{n+1}$ constitute leaf $u$ and axil $w$ -- note that the roles of $u,w$ are reversed if the other vertex is chosen as a special neighbour in the sense of Def.~\ref{def:Paulimeasurements}.

 If both leaf and axil are part of the target graph and constitute one of the target Bell pairs, i.e.,  $\{u,w\}\in\{\{a_1,a_2\},\{b_1,b_2\}\}$, this contradicts the assumption $a_1<b_1<a_2<b_2$, since there is no vertex between $u$ and $w$. Both leaf and axil can also not be part of the target graph while being in different Bell pairs:
 Leaf-axil pairs either remain such pairs under local complementations or turn into twins (cf. Theorem~\ref{thm:LCfoliage}, Fig.~\ref{fig:foliage} and Defs.~\ref{def:twin}, \ref{def:foliage}). 
This is a contradiction to $u$ and $w$ being in different Bell pairs of the target graph, since measuring the neighbourhood of axils or twins can never result in a graph with two connected components.
If just the axil $w$ is part of the target graph, Lemma \ref{lem:leafaxil} (a) with $v=u$ reduces the problem to $R_n$ and we can use our induction hypothesis.
If just the leaf $u$ is part of the target graph, we can use Lemma \ref{lem:leafaxil} (b) with $v=w$. 
\end{proof}

\begin{define}[Twin]\label{def:twin}
A twin is a vertex $v$ that has the same neighbourhood as a second vertex $w \neq v$ in the sense that
\begin{equation}
N_v \backslash \{w\} = N_w \backslash \{v\}. 
\end{equation}
\end{define}

\begin{define}[Foliage]\label{def:foliage}
The set containing all the leaves, axils and twins of a graph is called the foliage of that graph.
\end{define}

\begin{theorem}[Foliage is LC invariant]\label{thm:LCfoliage}
 The foliage of a graph $G$ is invariant under local complementation.
\end{theorem}

\begin{proof}
For $G=K_2=\tau_{v}(K_2)$ the statement is trivial as both vertices are leaves, axils, and twins at the same time.
For all other graphs note that a twin can be transformed into a leaf (or an axil) via local complementations: If twins $v$ and $w$ are neighbours, $\tau_w$ disconnects $v$ from all its other neighbours, that 
is, in $\tau_w(G)$ the vertex $v$ is a leaf and $w$ its axil.
If twins $v$ and $w$ are not neighbours, note that a twin pair  always has a common neighbour $u$ (unless $G=K_2$). In the graph $\tau_u(G)$, the two vertices are then neighbouring twins. With the above argument we know that in $\tau_w\circ \tau_u(G)$ vertex $v$ is a leaf and $w$ the corresponding axil.
Conversely, given a pair of leaf $v$ and axil $w$, the local complementation $\tau_w$ connects $v$ to every vertex in $N_w$, i.e., a twin pair is created. Again, choosing a common neighbour $u$ allows us to go to the graph $\tau_u\circ \tau_v(G)$ in which $w$ is a leaf and $v$ its axil. In fact, Fig.~\ref{fig:foliage} shows that all LCs on $u, v, w$ with $u\in N_v$ and/or $u\in N_w$ leave the foliage invariant. Since LCs can only transform the neighbourhood, we thus have shown that leaves, axils and twins can \emph{only} be transformed into each other with local complementations. 

As LCs are self-inverse this concludes the proof: Assume that a node in the foliage can be created via LCs out of a node that is not in the foliage. Then the reverse sequence of LCs would transform a node in the foliage to one that is not  -- contradicting the last sentence of the previous paragraph.
\end{proof}

\begin{figure*}[htpb]
\centering
	\includegraphics[width=1.0\textwidth,
    	trim={5.9cm 0.5cm 5.9cm 1.5cm}, 
    	clip
	]{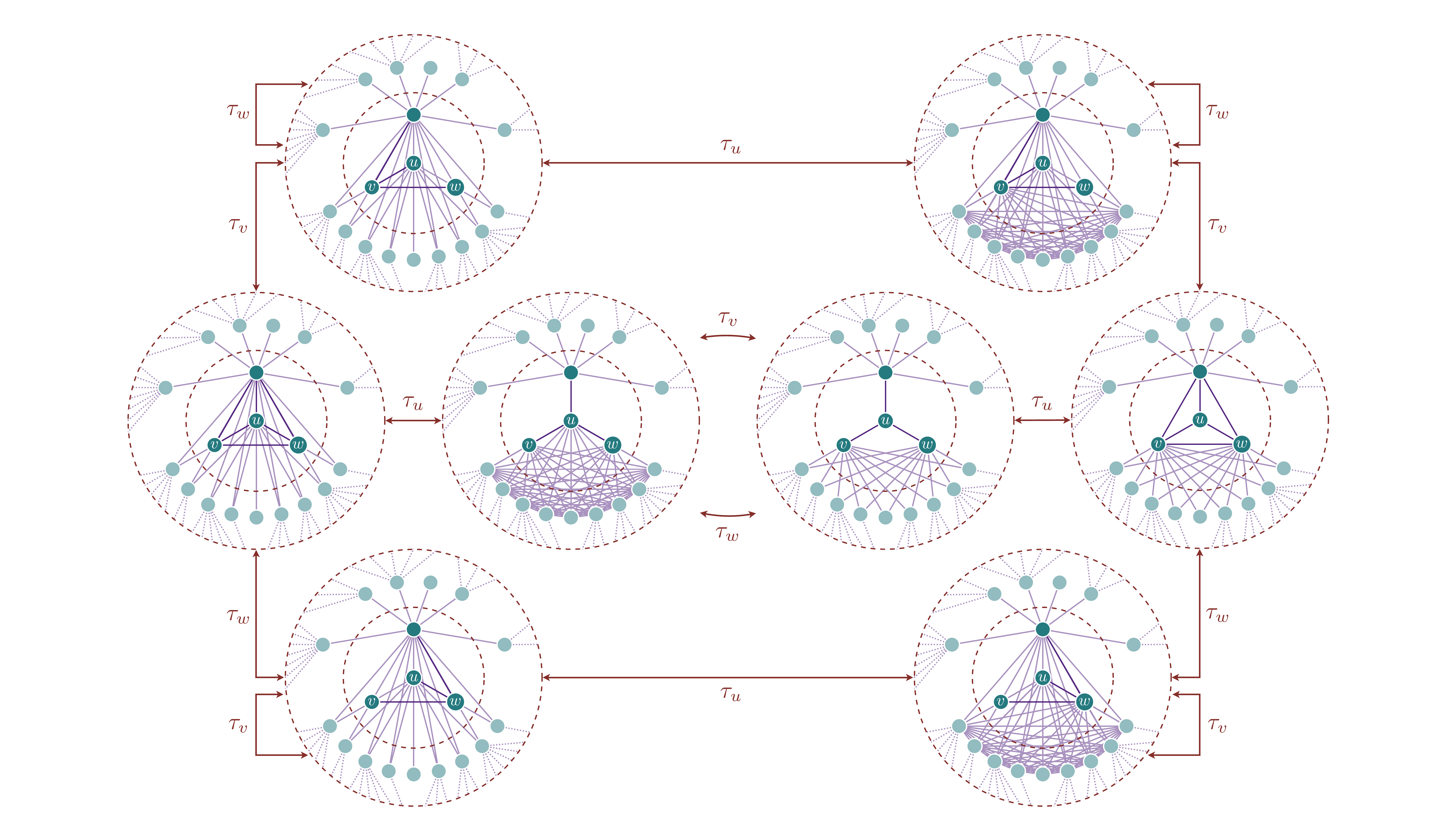}
		\caption{{\bf The foliage is LC invariant.} 
		The arrows indicate how foliage nodes $u, w$ transform under LCs. The central four graphs show twins $v, w$ that can be (dis-)connected via a LC with respect to one of their neighbours $u$.  The LCs $\tau_v$ and $\tau_w$ transform the outer graphs with the connected twins into graphs where either $v$ (upper row) or $w$ (lower row) is an axil and their former twin partner a leaf.}
		\label{fig:foliage}
\end{figure*}

\section{Line graphs}

In analogy to Theorem~\ref{thm:nocrossing}, we can  prove a no-go theorem for line graphs, i.e.,  $L_n\defeq(V_n,E_n)$ with vertices $V_n\defeq\{1,\ldots,n\}$ and edges $E_n\defeq\{(1,2),(2,3),\ldots,(n-1,n)\}$.

\begin{theorem}[No crossing on a line]\label{thm:nocrossingline}
It is not possible to extract two maximally entangled pairs from $|L_n\rangle$ if $a_1<b_1<a_2<b_2$ for any $n\in \mathbb{N}$ with local Clifford operations, local Pauli measurements and classical communication.
\end{theorem}

\begin{proof}
Again, the base case for $1 \leq n\leq 4$ is trivial and for $n=5$ and $n=6$ given by \cite{hahn2018quantum}. For the inductive step, we assume that Theorem~\ref{thm:nocrossingline} holds up to a given $n$. By the same argument as above, $K_2 \cup K_2$ is a vertex-minor of $L_{n+1}$ if and only if it is a vertex-minor of at least one of $X(L_{n+1})$, $Y(L_{n+1})$ or $Z(L_{n+1})$. 
We will now show that $K_2 \cup K_2\not < P_v(L_{n+1})$ for all $P_v \in \{X_v,Y_v,Z_v\}$, where we have $v \notin \{a_1,a_2,b_1,b_2\}$. Again, we consider three cases.\\

\begin{linie}[]\label{claim:4}
$K_2 \cup K_2\not < Z_v(L_{n+1})$.
\end{linie}

\noindent
If $v=1$ or $v=n+1$, we find $Z_v(L_{n+1})=L_n$ and can use our induction hypothesis. Otherwise, the $Z_v$-measurement splits the line into two line segments $L_i$ and $L_j$ with $i+j=n$. This implies $i,j \in \{1,2,\ldots, n-1\}$ and we can again use the induction hypothesis.\\

\begin{linie}[]\label{claim:5}
$K_2 \cup K_2\not < Y_v(L_{n+1})$.
\end{linie}

\noindent
We have $Y_v(L_{n+1})=L_n$, a graph on the $n$ vertices $\{1,2,\ldots,v-1,v+1,\ldots,n+1\}$, since local complementation with respect to $v$ connects $v-1$ to $v+1$ but leaves the remaining graph unchanged. Again, $K_2 \cup K_2$ cannot be a vertex-minor of $Y_v(L_{n+1})$ by our induction hypothesis.\\

\begin{linie}[]\label{claim:6}
$K_2 \cup K_2\not < X_v(L_{n+1})$.
\end{linie}

\noindent
If $v=1$ or $v=n+1$, we find $X_v(L_{n+1})=L_{n-1}$ and can use our induction hypothesis. Similarly, in the cases $v=2$ and $v=n$ we get $X_v(L_{n+1})=L_{n}$. In all other cases, we have $X_v(L_{n+1})$ equal to $L_{n-1}$ with an additional leaf, where the leaf-axil pair is made up by the set $\{v-1,v+1\}$. Using the same argument as in the proof of Ring~\ref{claim:3} -- involving Theorem~\ref{thm:LCfoliage} and Lemma \ref{lem:leafaxil} -- we can conclude our proof.\\
\end{proof}

\section{Outlook}

In this work, we build upon quantum network routing research to examine whether commonly used nearest-neighbour  architectures can aid with bypassing bottlenecks. We establish two no-go results for ring and line topologies. We show that, unlike the grid, whose smallest instance is the butterfly network, these two architectures are not suitable for bypassing bottlenecks without additional longer communication links.
This would be possible if we did allow for longer distance 2-local operations; as shown in Fig.~\ref{fig:ring_butterfly} we can indeed transform a ring graph state to a butterfly graph state and thereby enable the generation of crossing maximally entangled pairs.

Our investigation aims to increase our understanding of the potential and limitations of quantum network routing, in times when these settings are moving closer to experimental reality. The techniques that we have developed can potentially be used to eliminate more general quantum communication scenarios with existing bottlenecks. Relevant studies \cite{matsuo_analysis_2018} show how to generate the underlying graph states by sharing maximally entangled pairs between the nodes, while a proof-of-concept implementation using the IBM Quantum Experience has also been demonstrated \cite{pathumsoot_modeling_2020}. Finally, the long-distance communication links that are necessary in order to bypass the bottlenecks shown in this work, can in principle be built following the approach of Ref.~\cite{SMIKW16}. 
	
\section{Acknowledgements}
\noindent F.~H.~acknowledges support from the German Academic Scholarship Foundation, A.~P. from the German Research Foundation (DFG, Emmy Noether grant No. 418294583) and J.~E.~from the BMBF (Q.Link.X and QR.X), J.~E.~and A.~P.~ also acknowledge support from the Einstein Research Unit on Quantum Devices. This work has been initiated during the first Thematic Einstein Semester of the MATH+ Cluster of Excellence.

\begin{figure}[htbp]
\centering
	\includegraphics[width=0.46\textwidth]{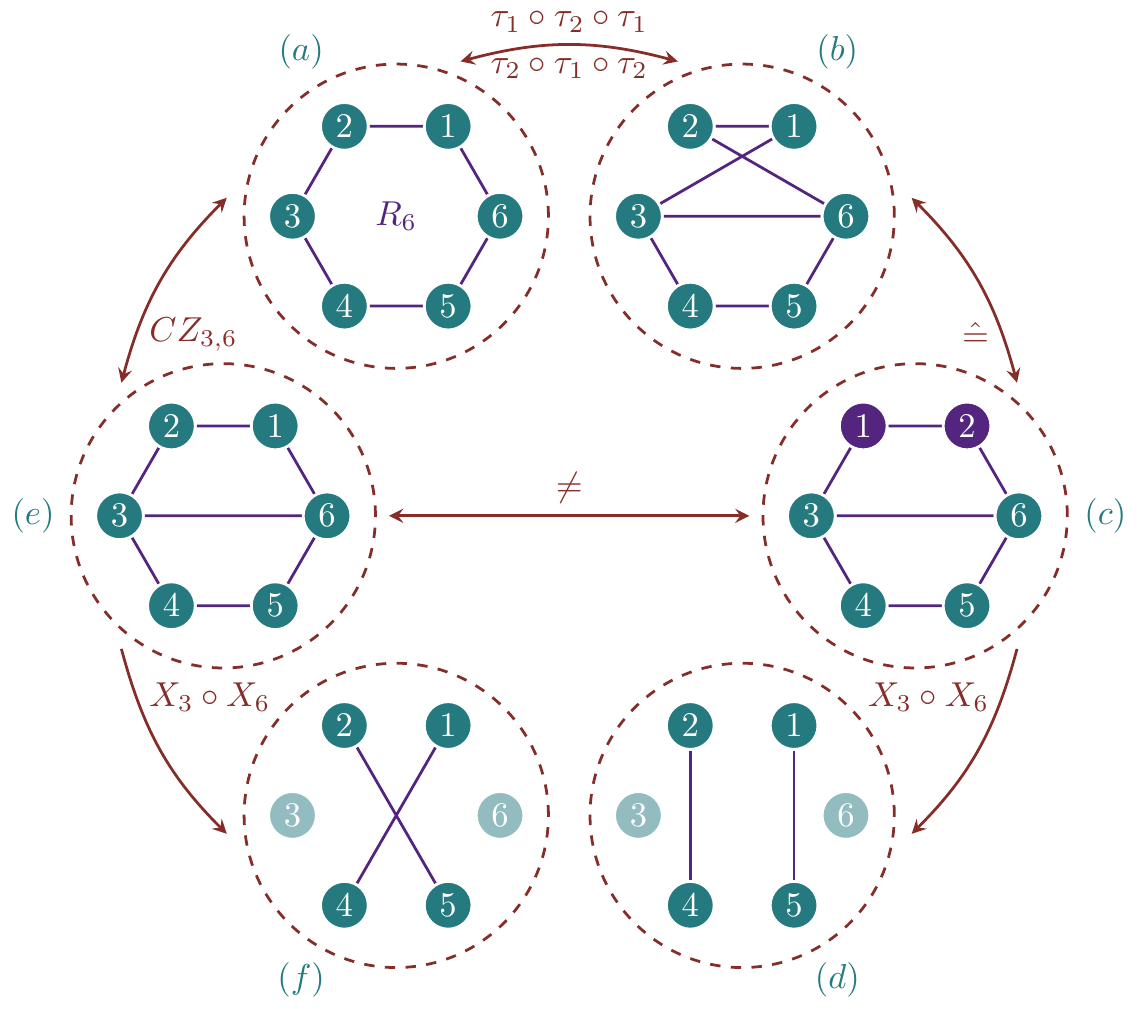}
		\caption{{\bf Ring $\mathbf{(a)}$ and butterfly $\mathbf{(e)}$ network.}
		While ring and butterfly networks are not LC-equivalent, a ring graph $(a)$ can be transformed via LCs into a butterfly-like graph $(c)$, where two nodes, here $1$ and $2$, are swapped ($\hat{=}$).
		The ring $R_6$ is transformed via a local complementation sequence $(a)$, $(b)\hat{=}(c)$. From 
		the resulting graph $(c)$, one can obtain maximally entangled pairs $(d)$ between nodes $(2,4)$ and  $(1,5)$ by measuring $3$ and $6$.
		 When allowing for 2-local operations (specifically $CZ$-gates) between nodes $3$ and $6$ $(e)$ one can obtain crossing Bell pairs $(1,4)$ and  $(2,5)$ $(f)$.}
		\label{fig:ring_butterfly}
\end{figure}

\bibliography{bib}

\begin{thebibliography}{24}
\expandafter\ifx\csname natexlab\endcsname\relax\def\natexlab#1{#1}\fi
\expandafter\ifx\csname bibnamefont\endcsname\relax
  \def\bibnamefont#1{#1}\fi
\expandafter\ifx\csname bibfnamefont\endcsname\relax
  \def\bibfnamefont#1{#1}\fi
\expandafter\ifx\csname citenamefont\endcsname\relax
  \def\citenamefont#1{#1}\fi
\expandafter\ifx\csname url\endcsname\relax
  \def\url#1{\texttt{#1}}\fi
\expandafter\ifx\csname urlprefix\endcsname\relax\def\urlprefix{URL }\fi
\providecommand{\bibinfo}[2]{#2}
\providecommand{\eprint}[2][]{\url{#2}}

\bibitem[{\citenamefont{Kimble}(2008)}]{QuantumInternet}
\bibinfo{author}{\bibfnamefont{H.~J.} \bibnamefont{Kimble}},
  \bibinfo{journal}{Nature} \textbf{\bibinfo{volume}{453}},
  \bibinfo{pages}{1023} (\bibinfo{year}{2008}).

\bibitem[{\citenamefont{Wehner et~al.}(2018)\citenamefont{Wehner, Elkouss, and
  Hanson}}]{QuantumInternetWehner}
\bibinfo{author}{\bibfnamefont{S.}~\bibnamefont{Wehner}},
  \bibinfo{author}{\bibfnamefont{D.}~\bibnamefont{Elkouss}}, \bibnamefont{and}
  \bibinfo{author}{\bibfnamefont{R.}~\bibnamefont{Hanson}},
  \bibinfo{journal}{Science} \textbf{\bibinfo{volume}{362}}
  (\bibinfo{year}{2018}).

\bibitem[{\citenamefont{Pant et~al.}(2019)\citenamefont{Pant, Krovi, Towsley,
  Tassiulas, Jiang, Basu, Englund, and Guha}}]{Guha}
\bibinfo{author}{\bibfnamefont{M.}~\bibnamefont{Pant}},
  \bibinfo{author}{\bibfnamefont{H.}~\bibnamefont{Krovi}},
  \bibinfo{author}{\bibfnamefont{D.}~\bibnamefont{Towsley}},
  \bibinfo{author}{\bibfnamefont{L.}~\bibnamefont{Tassiulas}},
  \bibinfo{author}{\bibfnamefont{L.}~\bibnamefont{Jiang}},
  \bibinfo{author}{\bibfnamefont{P.}~\bibnamefont{Basu}},
  \bibinfo{author}{\bibfnamefont{D.}~\bibnamefont{Englund}}, \bibnamefont{and}
  \bibinfo{author}{\bibfnamefont{S.}~\bibnamefont{Guha}}, \bibinfo{journal}{npj
  Quant. Inf.} \textbf{\bibinfo{volume}{5}}, \bibinfo{pages}{25}
  (\bibinfo{year}{2019}).

\bibitem[{\citenamefont{Meter et~al.}(2013)\citenamefont{Meter, Satoh, Ladd,
  Munro, and Nemoto}}]{VanMeter}
\bibinfo{author}{\bibfnamefont{R.~V.} \bibnamefont{Meter}},
  \bibinfo{author}{\bibfnamefont{T.}~\bibnamefont{Satoh}},
  \bibinfo{author}{\bibfnamefont{T.~D.} \bibnamefont{Ladd}},
  \bibinfo{author}{\bibfnamefont{W.~J.} \bibnamefont{Munro}}, \bibnamefont{and}
  \bibinfo{author}{\bibfnamefont{K.}~\bibnamefont{Nemoto}},
  \bibinfo{journal}{Networking Sc.} \textbf{\bibinfo{volume}{3}},
  \bibinfo{pages}{82} (\bibinfo{year}{2013}).

\bibitem[{\citenamefont{Schoute et~al.}(2016)\citenamefont{Schoute, Mancinska,
  Islam, Kerenidis, and Wehner}}]{SMIKW16}
\bibinfo{author}{\bibfnamefont{E.}~\bibnamefont{Schoute}},
  \bibinfo{author}{\bibfnamefont{L.}~\bibnamefont{Mancinska}},
  \bibinfo{author}{\bibfnamefont{T.}~\bibnamefont{Islam}},
  \bibinfo{author}{\bibfnamefont{I.}~\bibnamefont{Kerenidis}},
  \bibnamefont{and} \bibinfo{author}{\bibfnamefont{S.}~\bibnamefont{Wehner}}
  (\bibinfo{year}{2016}), \bibinfo{note}{arXiv:1610.05238}.

\bibitem[{\citenamefont{Hahn et~al.}(2019)\citenamefont{Hahn, Pappa, and
  Eisert}}]{hahn2018quantum}
\bibinfo{author}{\bibfnamefont{F.}~\bibnamefont{Hahn}},
  \bibinfo{author}{\bibfnamefont{A.}~\bibnamefont{Pappa}}, \bibnamefont{and}
  \bibinfo{author}{\bibfnamefont{J.}~\bibnamefont{Eisert}},
  \bibinfo{journal}{npj Quant. Inf.} \textbf{\bibinfo{volume}{5}},
  \bibinfo{pages}{1} (\bibinfo{year}{2019}).

\bibitem[{\citenamefont{Raussendorf et~al.}(2003)\citenamefont{Raussendorf,
  Browne, and Briegel}}]{raussendorf_measurement-based_2003}
\bibinfo{author}{\bibfnamefont{R.}~\bibnamefont{Raussendorf}},
  \bibinfo{author}{\bibfnamefont{D.~E.} \bibnamefont{Browne}},
  \bibnamefont{and} \bibinfo{author}{\bibfnamefont{H.~J.}
  \bibnamefont{Briegel}}, \bibinfo{journal}{Phys. Rev. A}
  \textbf{\bibinfo{volume}{68}}, \bibinfo{pages}{022312}
  (\bibinfo{year}{2003}).

\bibitem[{\citenamefont{de~Beaudrap and R{\"o}tteler}(2014)}]{debeaudrap14}
\bibinfo{author}{\bibfnamefont{N.}~\bibnamefont{de~Beaudrap}} \bibnamefont{and}
  \bibinfo{author}{\bibfnamefont{M.}~\bibnamefont{R{\"o}tteler}}, in
  \emph{\bibinfo{booktitle}{TQC 2014}} (\bibinfo{year}{2014}),
  vol.~\bibinfo{volume}{27}, pp. \bibinfo{pages}{217--233}.

\bibitem[{\citenamefont{Epping et~al.}(2016)\citenamefont{Epping, Kampermann,
  and Bru\ss}}]{EppingA}
\bibinfo{author}{\bibfnamefont{M.}~\bibnamefont{Epping}},
  \bibinfo{author}{\bibfnamefont{H.}~\bibnamefont{Kampermann}},
  \bibnamefont{and} \bibinfo{author}{\bibfnamefont{D.}~\bibnamefont{Bru\ss}},
  \bibinfo{journal}{New J. Phys.} \textbf{\bibinfo{volume}{18}},
  \bibinfo{pages}{053036} (\bibinfo{year}{2016}).

\bibitem[{\citenamefont{Leung et~al.}(2010)\citenamefont{Leung, Oppenheim, and
  Winter}}]{Butterfly}
\bibinfo{author}{\bibfnamefont{D.}~\bibnamefont{Leung}},
  \bibinfo{author}{\bibfnamefont{J.}~\bibnamefont{Oppenheim}},
  \bibnamefont{and} \bibinfo{author}{\bibfnamefont{A.}~\bibnamefont{Winter}},
  \bibinfo{journal}{IEEE Trans. Inf. Th.} \textbf{\bibinfo{volume}{56}},
  \bibinfo{pages}{3478} (\bibinfo{year}{2010}).

\bibitem[{\citenamefont{Hayashi et~al.}(2007)\citenamefont{Hayashi, Iwama,
  Nishimura, Raymond, and Yamashita}}]{hayashi_quantum_2007}
\bibinfo{author}{\bibfnamefont{M.}~\bibnamefont{Hayashi}},
  \bibinfo{author}{\bibfnamefont{K.}~\bibnamefont{Iwama}},
  \bibinfo{author}{\bibfnamefont{H.}~\bibnamefont{Nishimura}},
  \bibinfo{author}{\bibfnamefont{R.}~\bibnamefont{Raymond}}, \bibnamefont{and}
  \bibinfo{author}{\bibfnamefont{S.}~\bibnamefont{Yamashita}}, in
  \emph{\bibinfo{booktitle}{{STACS} 2007}}, edited by
  \bibinfo{editor}{\bibfnamefont{W.}~\bibnamefont{Thomas}} \bibnamefont{and}
  \bibinfo{editor}{\bibfnamefont{P.}~\bibnamefont{Weil}}
  (\bibinfo{publisher}{Springer}, \bibinfo{year}{2007}), pp.
  \bibinfo{pages}{610--621}.

\bibitem[{\citenamefont{Meignant et~al.}(2021)\citenamefont{Meignant,
  Grosshans, and Markham}}]{meignant2021classical}
\bibinfo{author}{\bibfnamefont{C.}~\bibnamefont{Meignant}},
  \bibinfo{author}{\bibfnamefont{F.}~\bibnamefont{Grosshans}},
  \bibnamefont{and} \bibinfo{author}{\bibfnamefont{D.}~\bibnamefont{Markham}},
  \bibinfo{journal}{arXiv preprint arXiv:2104.04745}  (\bibinfo{year}{2021}).

\bibitem[{\citenamefont{Hein et~al.}(2004)\citenamefont{Hein, Eisert, and
  Briegel}}]{Hein04}
\bibinfo{author}{\bibfnamefont{M.}~\bibnamefont{Hein}},
  \bibinfo{author}{\bibfnamefont{J.}~\bibnamefont{Eisert}}, \bibnamefont{and}
  \bibinfo{author}{\bibfnamefont{H.~J.} \bibnamefont{Briegel}},
  \bibinfo{journal}{Phys. Rev. A} \textbf{\bibinfo{volume}{69}},
  \bibinfo{pages}{062311} (\bibinfo{year}{2004}).

\bibitem[{\citenamefont{Hein et~al.}(2005)\citenamefont{Hein, D{\"u}r, Eisert,
  Raussendorf, Van~den Nest, and Briegel}}]{Hein06}
\bibinfo{author}{\bibfnamefont{M.}~\bibnamefont{Hein}},
  \bibinfo{author}{\bibfnamefont{W.}~\bibnamefont{D{\"u}r}},
  \bibinfo{author}{\bibfnamefont{J.}~\bibnamefont{Eisert}},
  \bibinfo{author}{\bibfnamefont{R.}~\bibnamefont{Raussendorf}},
  \bibinfo{author}{\bibfnamefont{M.}~\bibnamefont{Van~den Nest}},
  \bibnamefont{and} \bibinfo{author}{\bibfnamefont{H.-J.}
  \bibnamefont{Briegel}}, in \emph{\bibinfo{booktitle}{Proceedings of the
  Int.~School of Physics `Enrico Fermi'}} (\bibinfo{publisher}{IOS Press},
  \bibinfo{address}{Amsterdam}, \bibinfo{year}{2005}), pp. \bibinfo{pages}{115
  -- 218}.

\bibitem[{\citenamefont{Dahlberg
  et~al.}(2020{\natexlab{a}})\citenamefont{Dahlberg, Helsen, and
  Wehner}}]{dahlberg2020transforming}
\bibinfo{author}{\bibfnamefont{A.}~\bibnamefont{Dahlberg}},
  \bibinfo{author}{\bibfnamefont{J.}~\bibnamefont{Helsen}}, \bibnamefont{and}
  \bibinfo{author}{\bibfnamefont{S.}~\bibnamefont{Wehner}},
  \bibinfo{journal}{Quantum} \textbf{\bibinfo{volume}{4}}, \bibinfo{pages}{348}
  (\bibinfo{year}{2020}{\natexlab{a}}).

\bibitem[{\citenamefont{Dahlberg et~al.}(2022)\citenamefont{Dahlberg, Helsen,
  and Wehner}}]{dahlberg2022complexity}
\bibinfo{author}{\bibfnamefont{A.}~\bibnamefont{Dahlberg}},
  \bibinfo{author}{\bibfnamefont{J.}~\bibnamefont{Helsen}}, \bibnamefont{and}
  \bibinfo{author}{\bibfnamefont{S.}~\bibnamefont{Wehner}},
  \bibinfo{journal}{Inf. Proc. Lett.} \textbf{\bibinfo{volume}{175}},
  \bibinfo{pages}{106222} (\bibinfo{year}{2022}).

\bibitem[{\citenamefont{Dahlberg
  et~al.}(2020{\natexlab{b}})\citenamefont{Dahlberg, Helsen, and
  Wehner}}]{dahlberg2020}
\bibinfo{author}{\bibfnamefont{A.}~\bibnamefont{Dahlberg}},
  \bibinfo{author}{\bibfnamefont{J.}~\bibnamefont{Helsen}}, \bibnamefont{and}
  \bibinfo{author}{\bibfnamefont{S.}~\bibnamefont{Wehner}},
  \bibinfo{journal}{Quant. Sc. Tech.} \textbf{\bibinfo{volume}{5}},
  \bibinfo{pages}{045016} (\bibinfo{year}{2020}{\natexlab{b}}).

\bibitem[{\citenamefont{Kobayashi et~al.}(2010)\citenamefont{Kobayashi,
  Le~Gall, Nishimura, and R{\"o}tteler}}]{kobayashi2010perfect}
\bibinfo{author}{\bibfnamefont{H.}~\bibnamefont{Kobayashi}},
  \bibinfo{author}{\bibfnamefont{F.}~\bibnamefont{Le~Gall}},
  \bibinfo{author}{\bibfnamefont{H.}~\bibnamefont{Nishimura}},
  \bibnamefont{and}
  \bibinfo{author}{\bibfnamefont{M.}~\bibnamefont{R{\"o}tteler}}, in
  \emph{\bibinfo{booktitle}{2010 IEEE Int.~Symposium on Information Theory}}
  (\bibinfo{organization}{IEEE}, \bibinfo{year}{2010}), pp.
  \bibinfo{pages}{2686--2690}.

\bibitem[{\citenamefont{Van~den Nest
  et~al.}(2004{\natexlab{a}})\citenamefont{Van~den Nest, Dehaene, and
  De~Moor}}]{VandenNest1}
\bibinfo{author}{\bibfnamefont{M.}~\bibnamefont{Van~den Nest}},
  \bibinfo{author}{\bibfnamefont{J.}~\bibnamefont{Dehaene}}, \bibnamefont{and}
  \bibinfo{author}{\bibfnamefont{B.}~\bibnamefont{De~Moor}},
  \bibinfo{journal}{Phys. Rev. A} \textbf{\bibinfo{volume}{69}},
  \bibinfo{pages}{022316} (\bibinfo{year}{2004}{\natexlab{a}}).

\bibitem[{\citenamefont{Van~den Nest
  et~al.}(2004{\natexlab{b}})\citenamefont{Van~den Nest, Dehaene, and
  De~Moor}}]{VandenNest2}
\bibinfo{author}{\bibfnamefont{M.}~\bibnamefont{Van~den Nest}},
  \bibinfo{author}{\bibfnamefont{J.}~\bibnamefont{Dehaene}}, \bibnamefont{and}
  \bibinfo{author}{\bibfnamefont{B.}~\bibnamefont{De~Moor}},
  \bibinfo{journal}{Phys. Rev. A} \textbf{\bibinfo{volume}{70}},
  \bibinfo{pages}{034302} (\bibinfo{year}{2004}{\natexlab{b}}).

\bibitem[{\citenamefont{Oum}(2005)}]{VertexMinor}
\bibinfo{author}{\bibfnamefont{S.-I.} \bibnamefont{Oum}}, \bibinfo{journal}{J.
  Comb. Th. B} \textbf{\bibinfo{volume}{95}}, \bibinfo{pages}{79}
  (\bibinfo{year}{2005}).

\bibitem[{\citenamefont{Dahlberg and Wehner}(2018)}]{Wehner18}
\bibinfo{author}{\bibfnamefont{A.}~\bibnamefont{Dahlberg}} \bibnamefont{and}
  \bibinfo{author}{\bibfnamefont{S.}~\bibnamefont{Wehner}},
  \bibinfo{journal}{Phil. Trans. Roy. Soc. A} \textbf{\bibinfo{volume}{376}},
  \bibinfo{pages}{20170325} (\bibinfo{year}{2018}).

\bibitem[{\citenamefont{Matsuo et~al.}(2018)\citenamefont{Matsuo, Satoh,
  Nagayama, and Van~Meter}}]{matsuo_analysis_2018}
\bibinfo{author}{\bibfnamefont{T.}~\bibnamefont{Matsuo}},
  \bibinfo{author}{\bibfnamefont{T.}~\bibnamefont{Satoh}},
  \bibinfo{author}{\bibfnamefont{S.}~\bibnamefont{Nagayama}}, \bibnamefont{and}
  \bibinfo{author}{\bibfnamefont{R.}~\bibnamefont{Van~Meter}},
  \bibinfo{journal}{Phys. Rev. A} \textbf{\bibinfo{volume}{97}},
  \bibinfo{pages}{062328} (\bibinfo{year}{2018}).

\bibitem[{\citenamefont{Pathumsoot et~al.}(2020)\citenamefont{Pathumsoot,
  Matsuo, Satoh, Hajdušek, Suwanna, and Van~Meter}}]{pathumsoot_modeling_2020}
\bibinfo{author}{\bibfnamefont{P.}~\bibnamefont{Pathumsoot}},
  \bibinfo{author}{\bibfnamefont{T.}~\bibnamefont{Matsuo}},
  \bibinfo{author}{\bibfnamefont{T.}~\bibnamefont{Satoh}},
  \bibinfo{author}{\bibfnamefont{M.}~\bibnamefont{Hajdušek}},
  \bibinfo{author}{\bibfnamefont{S.}~\bibnamefont{Suwanna}}, \bibnamefont{and}
  \bibinfo{author}{\bibfnamefont{R.}~\bibnamefont{Van~Meter}},
  \bibinfo{journal}{Phys. Rev. A} \textbf{\bibinfo{volume}{101}},
  \bibinfo{pages}{052301} (\bibinfo{year}{2020}).

\end{thebibliography}

\end{document}